\documentclass[conference,a4paper]{IEEEtran} 
\usepackage{hyperref}
\hypersetup{final}

\usepackage[utf8]{inputenc} 
\usepackage[T1]{fontenc}

\usepackage[cmex10]{amsmath} 
\usepackage{amssymb}

\usepackage{flushend}
\usepackage{url}
\usepackage{cite}

\usepackage{microtype} 

 \usepackage[pdftex]{graphicx}

\def\F{{\ensuremath{\mathbb{F}}}}
\def\D#1#2{\frac{\mathrm{d}#1}{\mathrm{d}#2}} 
\DeclareMathOperator{\E}{\mathbb{E}} 
\let\sim\thicksim

\newtheorem{theorem}{Theorem}
\newtheorem{lemma}{Lemma}

\newtheorem{remark}{Remark}

\begin{document}
\title{Cumulant Expansion of Mutual Information\\
for Quantifying Leakage of a Protected Secret}

\author{%
	\IEEEauthorblockN{Olivier Rioul\IEEEauthorrefmark{1},
		Wei Cheng\IEEEauthorrefmark{1},
		and Sylvain Guilley\IEEEauthorrefmark{2}\IEEEauthorrefmark{1}}
	\IEEEauthorblockA{\IEEEauthorrefmark{1}%
		LTCI, T\'el\'ecom Paris, Institut Polytechnique de Paris,
		Palaiseau, France,
		firstname.lastname@telecom-paris.fr}
	\IEEEauthorblockA{\IEEEauthorrefmark{2}%
		Secure-IC S.A.S., 
		Tour Montparnasse, Paris, France,
		sylvain.guilley@secure-ic.com}
}

\maketitle

\begin{abstract}
The information leakage of a cryptographic implementation with a given degree of protection is evaluated in a typical situation when the signal-to-noise ratio is small. This is solved by expanding Kullback-Leibler divergence, entropy, and mutual information in terms of moments/cumulants.
\end{abstract}

\def\sigmaN{\sigma_{\!N}}
\def\sigmaZ{\sigma_{\!Z}}
\def\sigmaY{\sigma_{Y}}

\section{Introduction}

Consider the following threat model in any secrecy or privacy problem where the adversary guesses a secret (cryptographic key, password, identifier, etc.), modeled as a discrete random variable $X$, using some observation output of some \emph{side channel} (power consumption, electromagnetic emanation, acoustic noise, timing, etc.) modeled as a real-valued random variable~$Y$.
In side-channel applications targeting cryptographic implementations, the observation is generally made by some noisy measurement of a \emph{sensitive variable} $Z$, an unknown (possibly randomized) function of the secret $X$ which depends on the implementation. The noise is often modeled as Gaussian $N\sim\mathcal{N}(0,\sigmaN^2)$ independent of $(X,Z)$, and the observed $Y=Z+N$ is the output of an AWGN channel.
We are interested in how mutual information 
\begin{equation}
\label{equ:mi}
I(X;Y)=
h(Z+N) - h(Z+N\mid X)
\end{equation}
decreases as noise power $\sigmaN^2$ increases, that is, in a typical small signal-to-noise scenario.
The aim is to provide a theoretical leakage quantification as a dependency metric between secret~$X$ and attacker's observation~$Y$. This is particularly interesting for the designer who needs to evaluate the robustness of a given implementation to side-channel attacks.

In practice, the cipher algorithm is protected by some masking scheme in such a way that leakage is perfectly balanced at all orders $k<K$:
\begin{equation}\label{eqn-balanced_moments_Z_k}
\E(Z^k| X) = \E(Z^k) \quad\text{ a.s.} \quad (k=1,2,\ldots,K-1).
\end{equation}
Expanding powers $Y^k=(Z+N)^k$ and using the fact that $N$ is independent of $X$, it follows by induction that
\begin{equation}\label{eqn-balanced_moments_Y_k}
\E(Y^k| X) = \E(Y^k) \quad\text{ a.s.}  \quad (k=1,2,\ldots,K-1).
\end{equation}
The order $K$ is referred to as the \emph{high-order correlation immunity} (HCI) order by Carlet et al.~\cite{DBLP:journals/jce/CarletDGMP14}. It corresponds to the smallest moment of leakage that may depend on the secret. 
As a result, any attack from observation $Y$ based on correlation analysis of degree $k<K$ necessarily fails; $K$ is the minimal attack order that can succeed. 

The question now becomes: How does mutual information $I(X;Y)$ capture the fact that the $k$th order conditional moment $m_K(Y|X=x)=\E(Y^K|X=x)$ depends on $x$ when the noise increases?
Carlet et al.'s statement~\cite{DBLP:journals/jce/CarletDGMP14} is that $I(X;Y)$ is asymptotically $O(\sigmaN^{-2K})$ as $\sigmaN\to\infty$. This was taken as a fundamental result in the field of side-channel analysis. It was leveraged to illustrate the strength of  leakage squeezing~\cite[Fig.~4]{DBLP:journals/jmc/CarletDGM14},
to compare different countermeasures~\cite{DBLP:conf/eurocrypt/DucFS15,DBLP:conf/eurocrypt/GrossoS18}, and was extended in~\cite{DBLP:journals/tifs/ChengGCMD21} in the case of a code-based masking implementation where countermeasures can reduce mutual information by increasing the dual distance of the code and reducing its kissing number. 

Carlet et al.'s derivation~\cite{DBLP:journals/jce/CarletDGMP14}, however, is based on Cardoso's small cumulant approximation~\cite[Eq.\,(41)]{DBLP:journals/jmlr/Cardoso03} which in fact replaces Kullback-Leibler divergence by its quadratic approximation~\cite[Eq.\,(29)]{DBLP:journals/jmlr/Cardoso03}. As shown in this paper, this results in a problematic expansion of mutual information~\cite[Eq.\,(6)]{DBLP:journals/jce/CarletDGMP14}, which may yield ambiguous results. We make the appropriate corrections and find the asymptotic equivalent of $I(X;Y)$ up to $K=6$. Higher protection orders $K>6$ are rare in practice and involve cross-terms which make the asymptotic equivalent harder to find. Our main result is then the following\footnote{Throughout we use natural logarithms so that informational quantities are expressed in \emph{nats}.}. 
\begin{theorem}\label{thm:mi}
Let $X,Z$ be (discrete or continuous) real-valued random variables satisfying~\eqref{eqn-balanced_moments_Z_k} at orders $k=1,2,\ldots,K-1$ but \emph{not} at order $K$ (i.e., with at least one value $x$ such that $\E(Z^K| X=x)\ne \E(Z^K)$). Then 
if
$3\leq K\leq 6$,  
the following asymptotic equivalence holds  as $\sigmaN\to\infty$:
\begin{equation}\label{eq-final}
I(X;Y)\;\sim\;  \frac{ \mathrm{Var} \bigl(\E(Z^K|X)\bigr)}{2\cdot K! \cdot (\sigmaN^2+\sigmaZ^2)^K}
\end{equation}
where $\sigmaZ^2=\mathrm{Var}(Z)$ denotes variance and $\mathrm{Var} \bigl(\E(Z^K|X)\bigr)$ denotes  inter-class variance.
\end{theorem}

Our strategy to prove Theorem~\ref{thm:mi} is to rewrite mutual information in terms of \emph{non-Gaussianity} terms:
\begin{equation}
\begin{split}
\label{equ:mi:div}
I(X;Y) & 
= h(Y^*|X) - h(Y\mid X) -  \bigl(h(Y^*) - h(Y)\bigr) \\&= D(Y\|Y^*\mid X) - D(Y\|Y^*)
\end{split}
\end{equation}
where $Y^*$ is a Gaussian random variable independent of $X$ (hence $h(Y^*|X)=h(Y^*)$) with the same first and second order moments as~$Y$. 
We then go beyond the quadratic cumulant approximation of Cardoso~\cite[Eq.\,(29)]{DBLP:journals/jmlr/Cardoso03} and investigate how Kullback-Leibler divergences $D(Y\|Y^*)$ and $D(Y\|Y^*|X)$ behave as $\sigma$ increases, using a \emph{Gram-Charlier} expansion~\cite{hald2000early} in terms of a sequence of ``{modified moments}''.
This will fill the gap in proving Carlet et al.'s main result~\cite[Thm.\,1]{DBLP:journals/jce/CarletDGMP14}, while also giving the asymptotic equivalent for $3\leq K\leq 6$.
As we shall see, some annoying cross-terms prevent any straightforward generalization for $K>6$. 

The remainder of the paper is organized as follows. Section~\ref{sec:non-gau} reviews a kind of Gram-Charlier expansion and derives the corresponding non-Gaussianity expansions. Section~\ref{sec:exp:mi} gives the resulting expansions of mutual information and explains why the extension of~\eqref{eq-final} to $K>6$ is problematic. Numerical validation is carried out in Section~\ref{sec:num:simulation} in a practical code-based masking scheme in AES with Hamming weight leakage model. 
Section~\ref{sec:conclude} concludes.

\section{Cumulant Expansion of Non-Gaussianity}
\label{sec:non-gau}

Non-Gaussianity $D(Y\|Y^*)=h(Y^*)-h(Y)$ is a nonnegative quantity which vanishes if and only if $Y$ is Gaussian. 
For notational convenience write $\mu=\mu_Y$ and $\sigma=\sigmaY$.
Because $Y$ and $Y^*$ share the same mean $\mu$ and variance $\sigma$, it is convenient to write their densities in the form $\frac{1}{\sigma}f(\frac{y-\mu}{\sigma})$ and $\frac{1}{\sigma}g(\frac{y-\mu}{\sigma})$, respectively, where $f$ and $g$ are standardized densities (in particular $g=\mathcal{N}(0,1)$).
Since Kullback-Leibler divergence is invariant by invertible transformations, one has
\begin{equation}\label{eq:Dxx}
D(Y\|Y^*) \!=\! D\bigl(\tfrac{Y-\mu}{\sigma}\big\| \tfrac{Y^*-\mu}{\sigma}\bigr) \!=\! D(f\|g) \!=\! \int\!\!f \log \frac{f}{g}.
\end{equation}

\subsection{Density Expansion}
As $\sigmaN$ increases, $\sigma=\sqrt{\sigmaN^2+\sigmaZ^2}\to \infty$ but
high-order \emph{cumulants} $\kappa_3, \kappa_4, \ldots, \kappa_K$ of $Y$ remain bounded. In fact for $k\geq3$, $\kappa_k=\kappa_k(Y)=\kappa_k(Z)+\kappa_k(N)=\kappa_k(Z)$ are kept constant. On the other hand since $Y^*$ is Gaussian, all its high-order cumulants  are zero.
This, as we show in the next Lemma, can be used to show that the Gaussian noise~$N$ dominates in $Y=Z+N$ so that $f$ will approach the Gaussian $g$:
\begin{lemma}[Gram-Charlier Expansion]
\label{lem:gram:exp}
\begin{equation}\label{eq-GC}
\frac{f(x)}{g(x)}=1+\sum_{k=3}^K \frac{\widetilde{m}_k}{k!\sigma^k} H_k(x)+ o\bigl(\frac{1}{\sigma^K}\bigr) 
\end{equation}
where $H_k$ is the $k$-th Hermite polynomial ($H_3(x)=x^3-3x$, $H_4(x)=x^4-6x^2+3$,
$H_5(x)=x^5-10x^3+15x$, etc.) and where the ``modified moments''  $\widetilde{m}_k$ satisfy the recursion 
\begin{equation}\label{eq-recursion}
\widetilde{m}_k= \kappa_k +\sum_{j=3}^{k-3} \binom{k-1}{j} \widetilde{m}_j\kappa_{k-j}.
\end{equation}
\end{lemma}
\noindent The modified moments are computed exactly as the genuine moments $m_k$ are computed from the cumulants $\kappa_k$ using Smith's formula~\cite{smith1995recursive}, except that $\kappa_1$ and $\kappa_2$ are absent. Thus $\widetilde{m}_1=\widetilde{m}_2=0$, $\widetilde{m}_3=\kappa_3$, $\widetilde{m}_4=\kappa_4$, $\widetilde{m}_5=\kappa_5$, $\widetilde{m}_6=\kappa_6+10\kappa_3^2$, $\widetilde{m}_7=\kappa_7+35\kappa_3\kappa_4$, 
etc.
Notice that modified moments, like high-order cumulants, are bounded as $\sigma\to\infty$.
\begin{proof}
By definition of cumulants, the characteristic function $\phi_Y(t)=  \E(e^{itY})$  of $Y$ can be factorized as 
\begin{equation}\label{eq-factorcf}
\phi_Y(t)   
 =  \phi_{Y^*}\!(t) \,e^{\psi(t)}  
\end{equation}
where $\phi_{Y^*}\!(t) = e^{i\mu_Y t-\sigmaY^2 t^2/2}$ is the characteristic function of $Y^*\sim\mathcal{N}(\mu,\sigma^2)$ and  $\psi(t)
= \sum_{k=3}^K \kappa_k \frac{(it)^k}{k!}  +o(t^K)$.
Taking the exponential we expand $\exp\psi(t)= 1+\sum_{k=3}^K \widetilde{m}_k \frac{(it)^k}{k!} +o(t^K)$.
The coefficients $\widetilde{m}_k$ can be found by Taylor's formula and Leibniz's rule: $i^k\widetilde{m}_k=(e^\psi)^{(k)}(0)= (\psi' e^{\psi})^{(k-1)}(0)= \sum_{j}\!\binom{k-1}{j} (e^{\psi})^{(j)}(0)\psi^{(k-j)}\!(0)
$ which simplifies to~\eqref{eq-recursion}.
Now~\eqref{eq-factorcf} becomes
\begin{equation}\label{eq-expandcf}
\phi_Y(t) = \Bigl(1+\sum_{k=3}^K \frac{\widetilde{m}_k}{k!} (it)^k\Bigr)\phi_{Y^*}\!(t) +o(t^K) \phi_{Y^{\!*}}\!(t). 
\end{equation}
Taking the inverse Fourier transform gives the density of $Y$:
$$
\frac{1}{\sigma}f\bigl(\frac{y-\mu}{\sigma}\bigr)= \Bigl(1+\sum_{k=3}^K \frac{\widetilde{m}_k}{k!} (-\D{}{y})^k\Bigr) \frac{1}{\sigma}g\bigl(\frac{y-\mu}{\sigma}\bigr) + R(y)
$$
where we have used that multiplication by $(-it)$ in the Fourier domain (characteristic function) corresponds to differentiation.
Now by the defining property of Hermite polynomials,
\begin{equation*}
(-\D{}{y})^k  g\bigl(\frac{x-\mu}{\sigma}\bigr)  
= \frac{1}{\sigma^{k}} H_k\bigl(\frac{x-\mu}{\sigma}\bigr)\cdot g\bigl(\frac{x-\mu}{\sigma}\bigr).
\end{equation*}
The $o(t^K)$ term in~\eqref{eq-expandcf} having at most polynomial growth at infinity,
we can apply Watson's lemma~\cite[Chap.~2]{miller2006applied} for the remainder term $R(y)$, which gives $R(y)=o(\sigma^{-K})$ (with at most polynomial growth in $y$ at infinity). 
Letting $x=\frac{x-\mu}{\sigma}$ and dividing by $g(x)> 0$ gives the announced expansion.
\end{proof}


\begin{remark}\label{rem-0}
Contrary to what seems to be a popular belief in the literature (see e.g.,~\cite{DBLP:journals/jmlr/Cardoso03}), the coefficients multiplying the Hermite polynomials in the Gram-Charlier expansion~\eqref{eq-GC} are not just cumulants $\kappa_k$, but ``modified moments'' $\widetilde{m}_k$, which differ from cumulants as soon as $k\geq 6$.
\end{remark}

\subsection{Divergence Expansion}

\begin{theorem}\label{thm-div-expand-general}
The expansion of divergence in power of $\frac{1}{\sigma}$ is of the form
\begin{equation}
D(f\|g)=\sum_{k=3}^K \frac{c_k}{2k!\sigma^{2k}}  +o\Bigl(\frac{1}{\sigma^{2K}}\Bigr)
\end{equation}
where $c_k=\widetilde{m}^2_k$ + other terms of the form $\alpha_m \widetilde{m}_{k_1}\widetilde{m}_{k_2}\cdots \widetilde{m}_{k_m}$ where $m\geq 3$ and $k_1+k_2+\cdots+k_m=2k$.
\end{theorem}
 
\begin{proof}
Using~\eqref{eq-GC} in the form $\frac{f}{g}=1+h$ where $h=\sum_{k=3}^K \frac{\widetilde{m}_k}{k!\sigma^k} H_k(x)+ O({\sigma^{-K}})$, we proceed to expand $D(f\|g) =\int g (1+h) \log(1+h)$ where
$(1+h) \log(1+h) 
=h+\frac{h^2}{2}- \frac{h^3}{6}+\frac{h^4}{12}+\cdots +o(h^K)$.
Substituting gives
\begin{equation}
D(f\|g) \!=\! \int\!\!g h +\frac{1}{2}\!\int\!\! g h^2 - \frac{1}{6}\! \int\!\!g h^3 + \frac{1}{12}\! \int\!\!g h^4 +\cdots+ o\Bigl(\!\int\!\! g h^K\Bigr).
\end{equation}
By the orthogonality property of Hermite polynomials
\begin{equation}\label{eq-orth}
\int g H_k H_l  = k!\; \delta_{kl},
\end{equation}
one has $\int\!g H_k = \int\!g H_k H_0=  0$ ($k>0$) hence  $\int g h=0$. Moreover, by orthogonality,
$
\int g h^2 = \sum_{k=3}^K \bigl(\frac{\widetilde{m}_k}{k!\sigma^k} \bigr)^2 k! 
+o(\sigma^{-2K}) 
= \sum_{k=3}^K \frac{\widetilde{m}^2_k}{k!\sigma^{2k}} 
+o(\sigma^{-2K})
$. Thus the quadratic part $\frac{1}{2}\!\int\!\! g h^2$ accounts for the $\frac{\widetilde{m}^2_k}{2k!\sigma^{2k}}$ terms ($k\geq 3$).

The expansion of all higher-order terms $\int g h^m$ ($m\geq 3$) involve terms of the form $\frac{\widetilde{m}_{k_1}\widetilde{m}_{k_2}\cdots \widetilde{m}_{k_m}}{ \sigma^{k_1+k_2+\cdots+k_m}} \int g H_{k_1} H_{k_2} \cdots H_{k_m}$. Since each Hermite polynomial $H_k$ has the same parity as its degree $k$, all such terms vanish if $k_1+k_2+\cdots+k_m$ is odd. Hence there remains only terms in $\frac{1}{\sigma^{2k}}$ as stated.
\end{proof}

\begin{remark}
The asymptotic 
$
D(f\|g) = \frac{1}{2}\int g h^2 + o\bigl(\frac{1}{2}\int g h^2\bigr) 
$
was already proved in~\cite[Lemma~1]{DBLP:journals/tit/AbbeZ12}.
\end{remark}

\subsection{First Few Terms in the Divergence Expansion}

We can carry out the computations up to $K=6$.
The cubic and quartic terms can be evaluated at first orders using the special values~\cite[\S 6.8]{special}:  
$\int g H_3^2 H_4 = \frac{3! 3! 4!}{1! 2! 2!} = 216$, 
$\int g H_4^3 =  \frac{4!4!4!}{2!2!2!}= 1728$,
$\int g H_3 H_4 H_5 = \frac{5! 4! 3!}{3! 2! 1!} = 1440${\large\strut}, 
$\int g H_3^2 H_6 =\frac{3! 3! 6!}{0! 3! 3!} = 720$, and
$\int g H_3^4 = 3 \frac{3!^4}{0!^4 3!^2} + 6 \frac{3!^4}{0!^2 1!^2 2!^2} + \frac{3!^4}{1!^6}  = 3348$,
plus the fact that all terms in odd powers of $\sigma$ are zero (since 
they involve integrals $\int g H_kH_lH_m=0$ when $k+l+m$ is odd).
After some calculation we obtain
\begin{equation*}
\begin{split}
\int gh^3 
=\frac{648}{\sigma^{10}} \Bigl(\frac{ \widetilde{m}_3}{3!}\Bigr)^2\frac{\widetilde{m}_4}{4!}
+ \frac{1728}{\sigma^{12}} \Bigl(\frac{\widetilde{m}_4}{4!}\Bigr)^3
+ \frac{8640}{\sigma^{12}} \frac{\widetilde{m}_3}{3!}\frac{\widetilde{m}_4}{4!}\frac{\widetilde{m}_5}{5!}
\\+ \frac{2160}{\sigma^{12}} \Bigl(\frac{\widetilde{m}_3}{3!}\Bigr)^{2} \frac{\widetilde{m}_6}{6!}+O\Bigl(\frac{1}{\sigma^{14}}\Bigr)
\end{split}
\end{equation*}
and
\begin{equation*}
\int gh^4 
=\frac{3348}{\sigma^{12}}\Bigl(\frac{ \widetilde{m}_3}{3!}\Bigr)^4+O\Bigl(\frac{1}{\sigma^{14}}\Bigr).
\end{equation*}
Putting all pieces together and expressing modified moments in terms of cumulants, we obtain
\begin{equation}\label{eq-Rioul}
\begin{aligned}
D(f\|g)
&= \frac{ \widetilde{m}^2_3}{12\sigma^6} + \frac{ \widetilde{m}^2_4}{48 \sigma^8}
+ \frac{ \widetilde{m}^2_5}{240 \sigma^{10}}
-\frac{\widetilde{m}^2_3\widetilde{m}_4}{8\sigma^{10}} 
\!+\! \frac{ \widetilde{m}^2_6}{1440 \sigma^{12}}
\\&
- \frac{\widetilde{m}^3_4}{48\sigma^{12}} 
 - \frac{\widetilde{m}_3\widetilde{m}_4\widetilde{m}_5}{12\sigma^{12}} 
- \frac{\widetilde{m}^2_3 \widetilde{m}_6}{72\sigma^{12}} 
\!+\! \frac{31\widetilde{m}_3^4}{144\sigma^{12}} \!+\!O\Bigl(\frac{1}{\sigma^{14}}\Bigr)\\
&= \frac{ \kappa_3^2}{12\sigma^6} + \frac{ \kappa_4^2}{48 \sigma^8}
-\frac{\kappa_3^2\kappa_4}{8\sigma^{10}}+ \frac{ \kappa_5^2}{240\sigma^{10}} 
\\&+ \frac{7\kappa_3^4}{48\sigma^{12}}
- \frac{\kappa_4^3}{48\sigma^{12}} 
- \frac{\kappa_3\kappa_4\kappa_5}{12\sigma^{12}}
+ \frac{ \kappa_{6}^2}{1440\sigma^{12}}
\!+\!O\Bigl(\frac{1}{\sigma^{14}}\Bigr).
\end{aligned}
\end{equation}

\begin{remark}
In order to check the validity of~\eqref{eq-Rioul}, we can recover a known expression in a different model. Instead of having $Y=Z+N$, suppose that $Y=Y_1+Y_2+\cdots+Y_n$ where the $Y_i$'s are i.i.d. with mean $\mu$, variance $\sigma^2$, and high-order cumulants $\kappa_3$, $\kappa_4$, \ldots. 
The previous expansions can be used by replacing $\sigma$ by $\sqrt{n}\sigma$, $\kappa_k$ by $n\kappa_k$, and letting $n\to+\infty$.
The \emph{Gram-Charlier expansion}, re-ordered in powers of $\frac{1}{\sqrt{n}}$, becomes the \emph{Edgeworth expansion}
\begin{equation}
\begin{split}
\frac{f}{g}
%
\!=\!1\!+\!\frac{\kappa_3}{6\sigma^3\sqrt{n}}H_3+ \frac{\kappa_4}{24\sigma^4n}H_4
\!+\!\frac{\kappa_3^2}{72\sigma^6n}H_6
\!+\!\frac{\kappa_5}{120\sigma^5n\sqrt{n}}H_5 
\\\!+\!\frac{\kappa_{4}\kappa_{3} }{144\sigma^7n\sqrt{n}}H_7 
\!+\! \frac{ \kappa_3^3}{1296\sigma^9 n\sqrt{n}}H_9
\!+\! O\bigl(\frac{1}{n^2}\bigr).
\end{split}
\end{equation}
It is easily seen that all $O\bigl(\frac{1}{\sigma^{14}}\bigr)$ terms 
in~\eqref{eq-Rioul}  are then necessarily $O\bigl(\frac{1}{n^3}\bigr)$. Four terms out of the eight in~\eqref{eq-Rioul} are also $O\bigl(\frac{1}{n^3}\bigr)$, and there remains
\begin{equation}
D(f\|g)= \frac{ \kappa_3^2}{12n\sigma^6} + \frac{ \kappa_4^2}{48 n^2\sigma^8}
-\frac{\kappa_3^2\kappa_4}{8n^2\sigma^{10}}
+ \frac{7\kappa_3^4}{48n^{2}\sigma^{12}}
+O\Bigl(\frac{1}{n^3}\Bigr)
\end{equation}
which is exactly the result of Comon~\cite[Thm~14]{DBLP:journals/sigpro/Comon94} for his ``negentropy'' $D(f\|g)=h(g)-h(f)= \frac{1}{2}\log (2\pi e \sigma^2)-h(f)$. 
\end{remark}

\begin{remark}\label{rm-1}
The expansion~\eqref{eq-Rioul} contrasts with Cardoso's small cumulant approximation to the Kullback-Leibler divergence~\cite[Eq.~(41)]{DBLP:journals/jmlr/Cardoso03} which in our setting would read
$$
\frac{ \kappa_3^2}{12\sigma^6} + \frac{ \kappa_4^2}{48 \sigma^8}
+ \frac{ \kappa_5^2}{240\sigma^{10}} 
+ \frac{ \kappa_{6}^2}{1440\sigma^{12}}
+\cdots
$$
The difference with~\eqref{eq-Rioul} is due to two facts: (a) as already noticed in Remark~\ref{rem-0}, the coefficients of the \emph{Gram-Charlier expansion}~\eqref{eq-GC} are not the cumulants $\kappa_k$ for $k\geq 3$, but the modified moments $\widetilde{m}_k$, which differ from cumulants as soon as $k\geq 6$; (b) Cardoso's derivation only takes the quadratic approximation $\frac{1}{2}\int g h^2$ of divergence into account, ignoring higher order terms such as $\int g h^3=O(\frac{1}{\sigma^{10}})$. 

While (a) and (b) have no effect for the first two terms
$D(f\|g)= \frac{ \kappa_3^2}{12\sigma^6} + \frac{ \kappa_4^2}{48 \sigma^8}+O\bigl(\frac{1}{\sigma^{10}}\bigr) 
$,
both result in annoying higher-order cross-terms in the  genuine expression~\eqref{eq-Rioul} which do not appear in~\cite{DBLP:journals/jmlr/Cardoso03}. 
Because of this, derivations based on~\cite[Eq.~(41)]{DBLP:journals/jmlr/Cardoso03}, particularly the main result of~\cite{DBLP:journals/jce/CarletDGMP14}, become questionable as soon as $O\bigl(\frac{1}{\sigma^{10}}\bigr)$ terms are considered. 
\end{remark}

\begin{remark}
Since $D(f\|g)=D(Y\|Y^*)=h(Y^*)-h(Y)= \frac{1}{2}\log (2\pi e \sigma^2) -h(Y)$ we have the following expansion of (differential) entropy:
\vspace*{-3ex}
\begin{equation}
\begin{split}
h(Y)= 
\frac{1}{2}\!\log (2\pi e \sigma^2) 
- \frac{ \kappa_3^2}{12\sigma^6} 
 - \frac{ \kappa_4^2}{48 \sigma^8}
+\frac{\kappa_3^2\kappa_4}{8\sigma^{10}}
- \frac{ \kappa_5^2}{240\sigma^{10}} 
\\
- \frac{7\kappa_3^4}{48\sigma^{12}}
+ \frac{\kappa_4^3}{48\sigma^{12}} 
+\frac{\kappa_3\kappa_4\kappa_5}{12\sigma^{12}}
- \frac{ \kappa_{6}^2}{1440\sigma^{12}}
\!+\!O\Bigl(\frac{1}{\sigma^{14}}\Bigr).
\end{split}
\end{equation}
 
\end{remark}

\section{Cumulant Expansion of Mutual Information}
\label{sec:exp:mi}

\subsection{Mutual Information Expansion}

We now apply the expansion~\eqref{eq-Rioul} to both terms $D(Y\|Y^*)$ and $D(Y\|Y^*\mid X)$ in~\eqref{equ:mi:div}. To simplify the derivation we assume that $(K-1)$th order protection~\eqref{eqn-balanced_moments_Z_k} holds at least for the first two moments (hence $K\geq 3$): $\mu=\mu_Y=\mu_{Y|X=x}$ and $\sigma=\sigma_Y=\sigma_{Y|X=x}$ for all $x$. We can, therefore, apply \eqref{eq-Rioul} for $D(Y\|Y^*)$ and $D(Y\|Y^*\mid X=x)$ for a given secret value $x$, and then take the expectation over~$X$. Letting $\kappa_k(Z)=\kappa_k(Y)=\kappa_k$ and $\kappa_k(Z|X=x)=\kappa_k(Y|X=x)$ ($k\geq 3$) be the high-order cumulants of $Z$ and $Z|X=x$, respectively, we readily obtain
\begin{equation}\label{eq-Rioul2}
\resizebox{1.\hsize}{!}{$
\begin{split}
I(X;Y) =
\frac{\E\kappa_3^2(Z|X)-\kappa_3^2(Z)}{12\sigma^6} 
+ \frac{\E\kappa_4^2(Z|X)-\kappa_4^2(Z)}{48 \sigma^8} \qquad\quad
\\
-\frac{\E\bigl(\kappa_3^2(Z|X)\kappa_4(Z|X)\bigr)-\kappa_3^2(Z)\kappa_4(Z)}{8\sigma^{10}}
+ \frac{ \E\kappa_5^2(Z|X)-\kappa_5^2(Z)}{240\sigma^{10}} 
\\
 + \frac{7\bigl(\E\kappa_3^4(Z|X)-\kappa_3^4(Z)\bigr)}{48\sigma^{12}} 
- \frac{\E\kappa_4^3(Z|X)-\kappa_4^3(Z)}{48\sigma^{12}} 
\\
- \frac{\E\bigl(\kappa_3(Z|X)\kappa_4(Z|X)\kappa_5(Z|X)\bigr)-\kappa_3(Z)\kappa_4(Z)\kappa_5(Z)}{12\sigma^{12}}
\\
+ \frac{ \E\kappa_{6}^2(Z|X)-\kappa_{6}^2(Z)}{1440\sigma^{12}}
+O\Bigl(\frac{1}{\sigma^{14}}\Bigr). \qquad\qquad
\end{split} $}
\end{equation}

\begin{remark}
This contrasts with the high-order expansion of mutual information in~\cite[Eq.~(6)]{DBLP:journals/jce/CarletDGMP14} which reads
\begin{equation*}
\begin{split}
I(X;Y)\!=\!
\frac{\E\bigl(\kappa_3(Z|X)\!-\!\kappa_3(Z)\bigr)^2}{12\sigma^6} 
\!+\! \frac{\E\bigl(\kappa_4(Z|X)\!-\!\kappa_4(Z)\bigr)^2}{48 \sigma^8} 
\\
\!+\! \frac{ \E\bigl(\kappa_5(Z|X)\!-\!\kappa_5(Z)\bigr)^2}{240\sigma^{10}} 
\!+\! \frac{ \E\bigl(\kappa_{6}(Z|X)\!-\!\kappa_{6}(Z)\bigr)^2}{1440\sigma^{12}}
\!+\!O\Bigl(\frac{1}{\sigma^{14}}\Bigr) . 
\end{split}
\end{equation*}
The difference with~\eqref{eq-Rioul2} is due to three facts: (a) and (b) leading to annoying cross-terms in the non-Gaussianity expansion, as explained in Remark~\ref{rm-1}; (c) terms of the form $\E\kappa_k^2(Z|X)-\kappa_k^2(Z)$ can be written as variances
\begin{equation}
\label{equ:exp:2:var}
\E\kappa_k^2(Z|X)-\kappa_k^2(Z)
=
\E \bigl(\kappa_k(Z|X)-\kappa_k(Z)\bigr)^2
\end{equation}
only under the condition that $\kappa_k(Z)=\E\kappa_k(Z|X)$. This condition indeed holds for $k=3,4,5$ under the above assumptions because of the well-known expressions of $\kappa_3$, $\kappa_4$, and $\kappa_5$ in terms of moments $m_1$, $m_2$, $m_3$, $m_4$, $m_5$, where the quantities  $m_1(Z|X=x)=\E(Z|X=x)=\E(Z)=m_1(Z)$ and $m_2(Z|X=x)=\E(Z^2|X=x)=\E(Z^2)=m_2(Z)$ do not depend on $X=x$ and where
$m_k(Z)=\E(Z^k)=\E\,\E(Z^k|X)=\E m_k(Z|X)$. However,  the condition $\kappa_k(Z)=\E\kappa_k(Z|X)$ is no longer satisfied for $k=6$ because of the  $-10m_3^2$ term in the expression of $\kappa_6=  m_{6}-6m_{5}m_{1}-15m_{4}m_{2}+30m_{4} m_{1}^{2}-10 m_{3}^{2}+120m_{3}m_{2}m_{1}-120m_{3} m_{1}^{3}+30 m_{2}^{3} -270 m_{2}^{2} m_{1}^{2}+360m_{2} m_{1}^{4}-120 m_{1}^{6}$.
 
 While (a), (b), and (c) have no effect for the first two terms $I(X;Y)=
\frac{\E(\kappa_3(Z|X)\!-\!\kappa_3(Z))^2}{12\sigma^6} 
+ \frac{\E(\kappa_4(Z|X)\!-\!\kappa_4(Z))^2}{48 \sigma^8}+O\bigl(\frac{1}{\sigma^{10}}\bigr)$,
they result in annoying higher-order cross-terms in the  genuine expression~\eqref{eq-Rioul2} which do not appear in~\cite{DBLP:journals/jce/CarletDGMP14}.
\end{remark}

\begin{proof}[Proof of the main Theorem~\ref{thm:mi}]
The HCI condition~\eqref{eqn-balanced_moments_Z_k} states that $m_k(Z|X)=m_k(Z)$ a.s. for $k<K$.
Now from the well-known formulas expressing cumulants in terms of moments, one has
$\kappa_k(Z|X)=m_k(Z|X)+$ lower-order terms in $m_1(Z|X)=m_1(Z),\ldots,m_{k-1}(Z|X)=m_{k-1}(Z)$.
It follows that $\kappa_k(Z|X)=\kappa_k(Z)$ a.s. for $k<K$ while for $k=K$, we have $\kappa_K(Z|X)-m_K(Z|X)=\kappa_K(Z)-m_K(Z)$.
Thus, $\kappa_K(Z|X) -\kappa_K(Z)=m_K(Z|X) -m_K(Z)$  and in particular 
$\E\kappa_K(Z|X) -\kappa_K(Z)=\E m_K(Z|X) -m_K(Z) = \E\,\E(Z^K|X) - \E(Z^K)=0$. Therefore, we can write
$ \E\kappa^2_K(Z|X)-\kappa^2_K(Z) = \mathrm{Var} \bigl(\kappa_K(Z|X)\bigr)
=E\bigl(\kappa_K(Z|X)-\kappa_K(Z)\bigr)^2=\E\bigl(m_K(Z|X)-m_K(Z)\bigr)^2=\mathrm{Var} \bigl(m_K(Z|X)\bigr)=\mathrm{Var} \bigl(\E(Z^K|X)\bigr)$ which is nonzero since $\E(Z^K|X)$ is not constant a.s. 

By examination of~\eqref{eq-Rioul2} when $K\leq 6$, it is easily seen that
\begin{equation}\label{eqn-IXY_2K}
\begin{aligned}
I(X;Y) &= \frac{\E\kappa^2_K(Z|X)-\kappa^2_K(Z)}{2 \cdot K!\cdot \sigma^{2K}} +o\Bigl(\frac{1}{\sigma^{2K}}\Bigr)\\&=\frac{ \mathrm{Var} \bigl(\E(Z^K|X)\bigr)}{2\cdot K! \cdot \sigma^{2K}}+o\Bigl(\frac{1}{\sigma^{2K}}\Bigr)
\end{aligned}
\end{equation}
where $\sigma^2=\sigma_Y^2=\sigmaN^2+\sigmaZ^2$. 
\end{proof}

\begin{remark}
What makes the proof of Theorem~\ref{thm:mi} work in that in~\eqref{eq-Rioul}, all terms in $\frac{1}{\sigma^{2k}}$ ($k=3,4,5,6$) involve only cumulants of order $\leq k$. 

This property, however, does not generalize to higher orders.
In fact by Theorem~\ref{thm-div-expand-general}, there is at least one additional term in $\frac{1}{\sigma^{14}}$ in the form $\alpha_3 \widetilde{m}_{8}\widetilde{m}^2_{3}$ (since $8+3+3=14$) which will contribute to a  term
$\frac{\alpha_3 \kappa^2_3(Z)(\E\kappa_8(Z|X)-\kappa_8(Z))}{\sigma^{14}}$
in addition of the $\frac{\E\kappa^2_7(Z|X)-\kappa^2_7(Z)}{10080\cdot \sigma^{14}}$ of~\eqref{eqn-IXY_2K}.
Assuming $\kappa_3(Z)\ne 0$, we still have $I(X;Y)=O(\sigma^{-2K})$ for $K=7$  but with a different asymptotic equivalent.

Furthermore, again assuming $\kappa_3(Z)\ne 0$, for $K=8$ the term $\frac{\alpha_3 \kappa^2_3(Z)(\E\kappa_8(Z|X)-\kappa_8(Z))}{\sigma^{14}}$ still contributes to mutual information so that in the case it is no longer true that $I(X;Y)=O(\sigma^{-2K})$.
We still have $I(X;Y)=O(\sigma^{-14})$ instead of $I(X;Y)=O(\sigma^{-16})$.

In general for higher orders, the terms $\alpha_m \widetilde{m}_{k_1}\widetilde{m}_{k_2}\cdots \widetilde{m}_{k_m}$  ($m\geq 3$, $k_i\geq 3$, $k_1+k_2+\cdots+k_m=2k$) of Theorem~\ref{thm-div-expand-general} 
will not contribute to $I(X;Y)$ only when all $k_i$ are necessarily $<K$. Since the maximum possible $k_i$ is $2k-6$ (for $m=3$, the other two $k_i$'s being equal to $3$), we must have at least~\eqref{eqn-balanced_moments_Z_k} satisfied at order $2k-5$ to ensure that $I(X;Y)=O(\sigma^{-2(k+1)})$.
Therefore, for $K\geq 7$, $I(X;Y)=O(\sigma^{-2K})$ requires an HCI at least  $2K-7$.

In practice, such extremely high-order protection ($K=5,6,7,\ldots$) is unthinkable for all implementations. Hence Theorem~\ref{thm:mi} will apply to all cases of interest. In the following section we illustrate this using a code-based masked implementation for $K\leq 4$.
\end{remark}

\section{Numerical Simulations}
\label{sec:num:simulation}

Consider an advanced encryption standard (AES~\cite{website-fips197}) block cipher, which takes in input a plaintext of $16$ bytes, and outputs a ciphertext of the same size.
The attacker is able to monitor inputs and outputs, but does not know the secret key.
In such a cryptographic algorithm, it is practically impossible to deduce the secret from inputs and outputs: all the security relies on the secrecy of the key, in keeping with Kerckhoffs's principle~\cite{kerckhoffs_2} (a.k.a. Shannon's maxim~\cite{shannon-onetimepad}).

\begin{figure}[h!]
\centering
\includegraphics[width=1.0\linewidth]{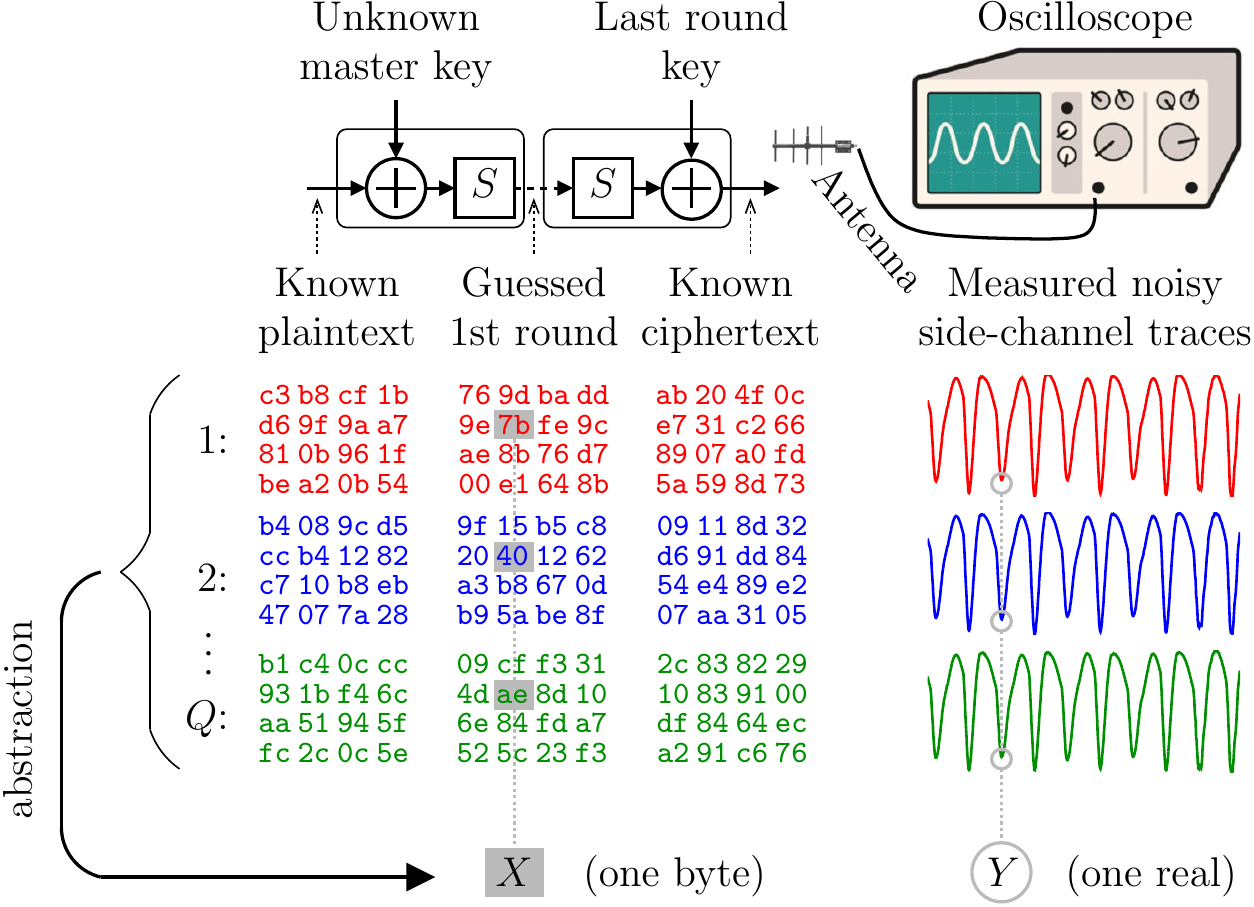}
\vspace{-0.2cm}
\caption{Public information (plaintext and ciphertext) available to an attacker and first-round key-dependent intermediate (discrete) value, put in front of corresponding side-channel execution traces (analog) observed by the attacker.}
\label{fig-sca_example-crop}
\end{figure}

Side-channel attacks consist in measuring power consumption~\cite{kocher-dpa_and_related_attacks} or  electromagnetic (EM) waves~\cite{Gandolfi:2001:EAC:648254.752700} produced during the execution of the AES algorithm. As shown in Fig.~\ref{fig-sca_example-crop}, the attacker measures waveforms corresponding to the side-channel emanation of the AES computation.
Such side information (repeatedly collected many times) is correlated to the secret key, and the attacker tries to exploit it in order to validate assumptions on small chunks of the key. In Fig.~\ref{fig-sca_example-crop},  the reference $128$-bit key is \texttt{0x2b7e151628aed2a6abf7158809cf4f3c} (taken as an example in the NIST specification \cite[Appendix~A]{website-fips197}) and the guessed values are that of the first round of AES, which consists in the application of AES SubBytes on the plaintext XORed with the key.
The measured waveforms are time series of power or EM emanations, which depend on the plaintext (or equivalently, on the ciphertext, since encryption is symmetric). Some specific samples depend on small chunks of the plaintext/ciphertext and of the secret key, and are used by the attacker to assess hypotheses on the key. 

We consider a practical case where the block cipher algorithm is protected by a masking scheme~\cite{DBLP:conf/eurocrypt/ProuffR13}. The order of protection rarely exceeds $K=4$.
Specifically, we target a two-share masking scheme~\cite{DBLP:journals/tifs/ChengGCMD21} in which the key chunk $X\in\F_{q}$ is encoded as $\left(X\oplus (C\otimes M),\, M\right)$, using an independent uniformly distributed mask $M\in\F_{q}$ and a  nonzero constant $C\in\F_{q}$.  Here $\oplus$ and $\otimes$ denote the addition and multiplication, respectively, in a finite field $\F_{q}=\F_{16}$ or $\F_{256}$ 
\footnote{The irreducible polynomial we used in this paper are $\alpha^4 + \alpha + 1$ for $\F_{16}$ and $ \alpha^8 + \alpha^4 + \alpha^3 + \alpha^2 + 1$ for $\F_{256}$.}.

Hamming weight model + Gaussian noise is a commonly used model in side-channel analysis.
The leaked sensitive variable is modeled as a $Z=w_H\left(X\oplus (C\otimes M)\right) + w_H(M)$ where $w_H(\cdot)$ denotes the Hamming weight (number of nonzero bits) and $Y=Z+N$ where $N\sim\mathcal{N}(0,\sigmaN^2)$.
As demonstrated in~\cite{DBLP:journals/tifs/ChengGCMD21} and shown in Table~\ref{tab:C:Ex}, both HCI $K$ and $\mathrm{Var}\bigl(\E(Z^K|X)\bigr)$ change with different choices of $C$.  We can, therefore, validate Theorem~\ref{thm:mi} in multiple cases. 

\begin{table}[h!]
\vspace{-0.2cm}
\renewcommand{\arraystretch}{1.2}
\centering
\caption{Different $K$ and $V_k=\mathrm{Var}\bigl(\E(Z^k|X)\bigr)$ $(k=1,\ldots,K)$ by using different $C$ (in decimal representation).} \label{tab:C:Ex}
\resizebox{0.49\textwidth}{!}{
\begin{tabular}{|c|c|c|c|c||c|c|c|c|c|}
\cline{2-10}
\multicolumn{1}{c|}{} & \multicolumn{4}{|c||}{$X\in\F_{16}$} & \multicolumn{5}{c|}{$X\in\F_{256}$} \\\hline
$C$ & 1&4&8&3&1&128&143&45&29 \\
$K$ & 2&2&2&3&2&2&3&3&4\\\hline
$V_1$ & 0&0&0&0&0&0&0&0&0\\
$V_2$ & 1&0.5&0.25&0&2.0&0.5&0&0&0\\
$V_3$ & - & -& -&0.25&-&-&3.9375&0.5625&0\\
$V_3$ & - & - & - & - & - & - & - & - & 6.75\\\hline
\end{tabular}}
\end{table}

The numerical results of mutual information are shown in Fig.~\ref{fig:mi} in log-log scale, where
slope $-K$ indicate $I(X;Y)\sim \mathrm{Cst}\cdot \sigmaN^{-2K}$.
We observe the first nonzero order expansion of mutual information dominates when the noise level is high enough (e.g., when $\sigmaN^2 \geq 10$). Overall Theorem~\ref{thm:mi} gives an accurate approximation of mutual information.

\begin{figure}[!h]
	\vspace{-0.6cm}
	\centering
		\includegraphics[width=1.03\linewidth]{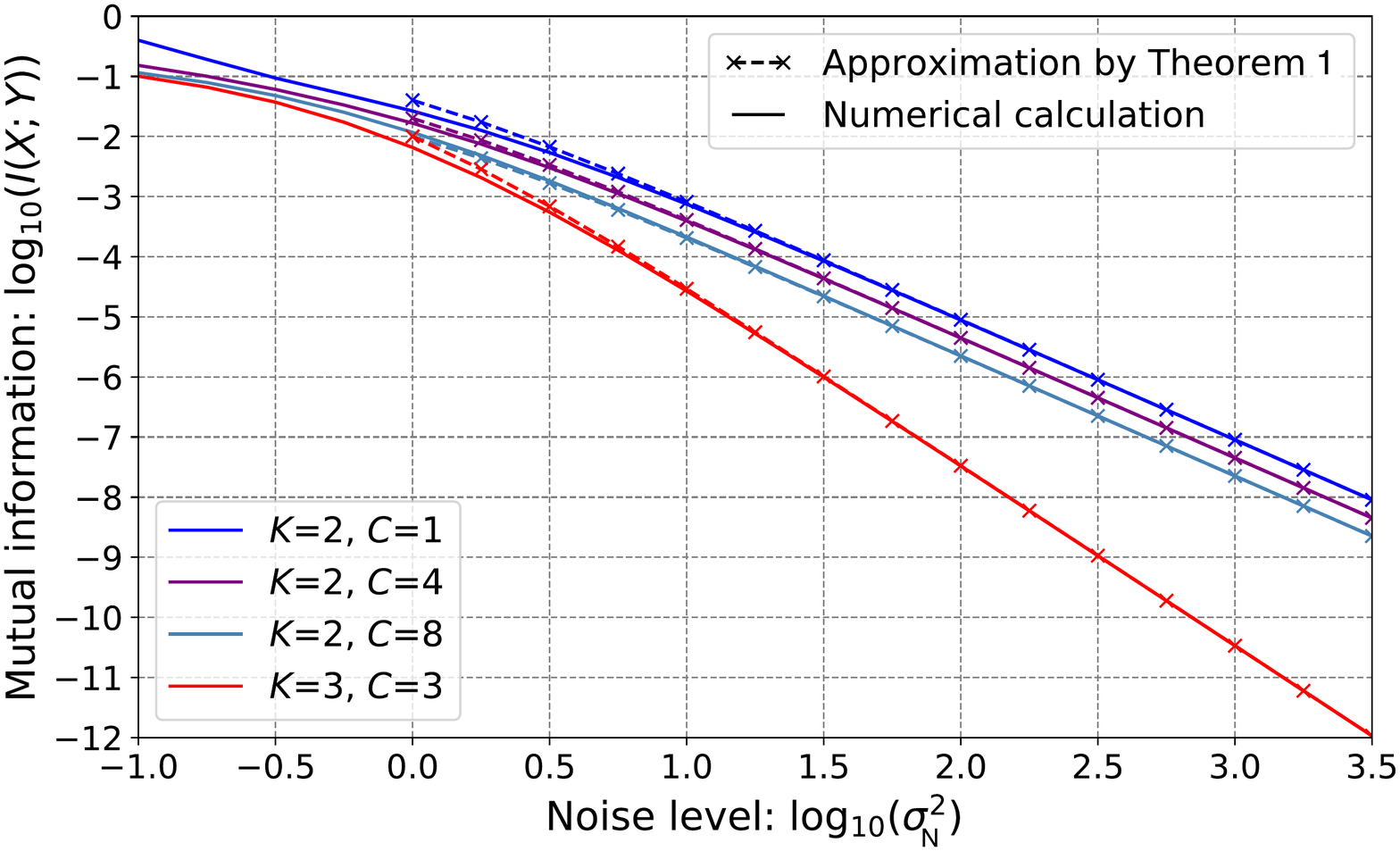}
\vspace*{-1.4cm}\\
		\includegraphics[width=1.03\linewidth]{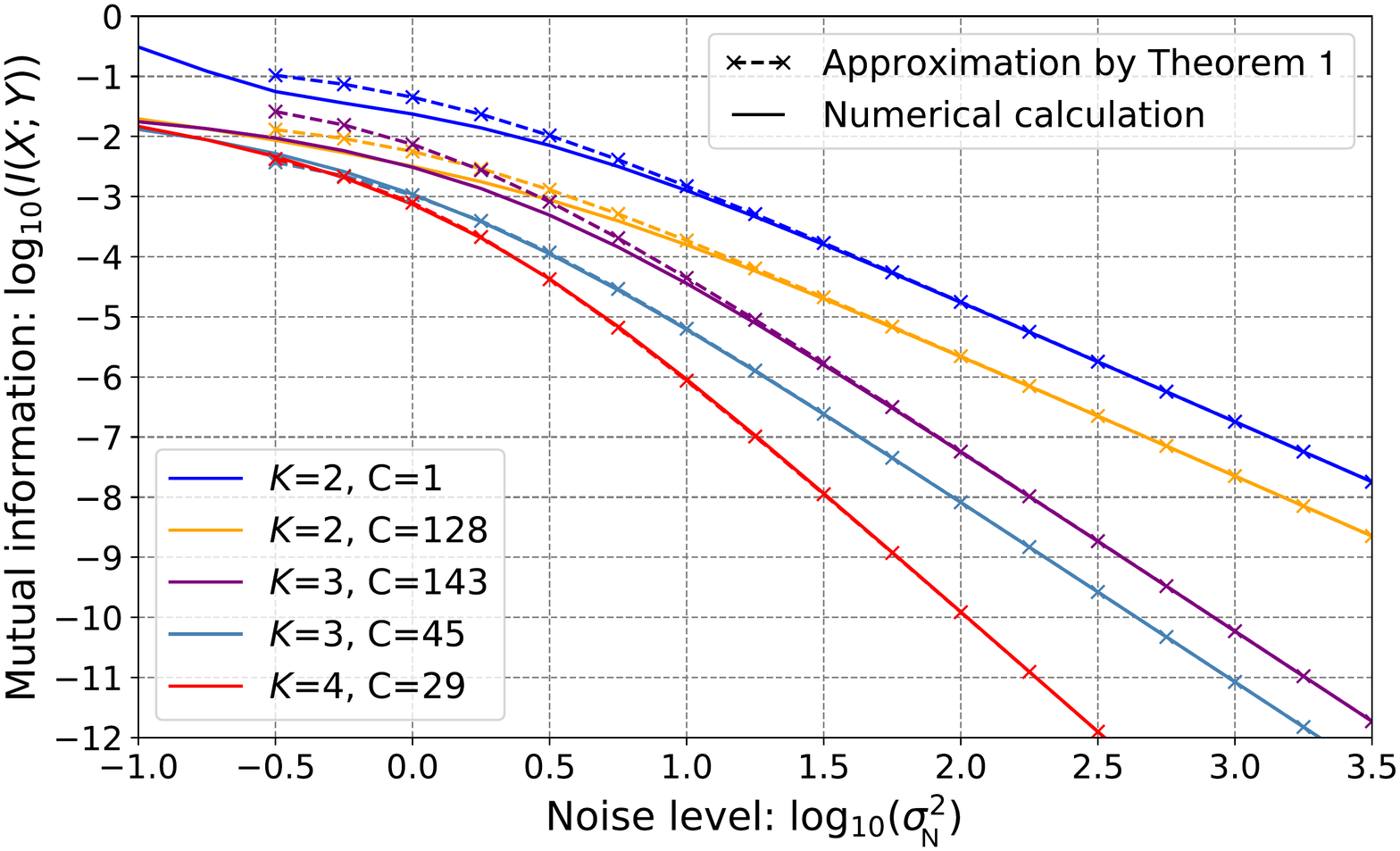}
\vspace*{-1.2cm}
	\caption{Numerical validation of Theorem~\ref{thm:mi} by taking different $C$ therefore different $K$ and $\mathrm{Var}\bigl(\E(Z^K|X)\bigr)$ (a) $X\in\F_{16}$. (b) $X\in\F_{256}$. }
	\label{fig:mi}
	\vspace*{-.4cm}
\end{figure}

\section{Conclusion}
\label{sec:conclude}

In this paper, we presented a cumulant-based expansion of  Kullback-Leibler divergence and mutual information with application to side-channel analysis. We fixed the mathematical issue that existed in the literature and proposed a rigorous proof for the main result in~\cite{DBLP:journals/jce/CarletDGMP14} in most cases of interest.

\providecommand{\noopsort}[1]{}

\end{document}